\documentclass[a4paper]{article}[11pt]

\usepackage{amsmath}
\usepackage{amssymb}
\usepackage{amsthm}
\usepackage{graphicx}
\usepackage[colorinlistoftodos]{todonotes}
\usepackage{algorithm}
\usepackage[noend]{algpseudocode}
\usepackage{enumerate}
\usepackage{xcolor}
\usepackage{framed}
\usepackage{bm}
\usepackage{pifont}
\usepackage{multicol}
\usepackage{lipsum}
\usepackage{tikz}
\usetikzlibrary{calc}
\usepackage{wrapfig}
\usepackage{array}
\usepackage{relsize}
\usepackage{comment}
\usepackage{url}
\usepackage[title]{appendix}
\usepackage{caption}
\usepackage{subcaption}
\usepackage{xstring}
\usepackage{times}

\usepackage[english]{babel}
\usepackage[utf8x]{inputenc}
\usepackage[T1]{fontenc}

\usepackage[a4paper,top=3cm,bottom=2cm,left=3cm,right=3cm,marginparwidth=1.75cm]{geometry}

\usepackage[colorlinks=true, allcolors=blue]{hyperref}

\newtheorem{observation}{Observation}

\newtheorem{lemma}{Lemma}
\newtheorem{corollary}[lemma]{Corollary}

\newtheorem{definition}{Definition}

\newtheorem{theorem}{Theorem}

\algdef{SE}[RECEIVING]{Receiving}{EndReceiving}[1]{\textbf{upon
receiving}\ #1\ \algorithmicdo}{\algorithmicend\ \textbf{}}%
\algtext*{EndReceiving}

\algdef{SE}[UPON]{Upon}{EndUpon}[1]{\textbf{upon
}\ #1\ \algorithmicdo}{\algorithmicend\ \textbf{}}%
\algtext*{EndUpon}

\newcommand{\BCname}{4-Stage-f+1-Provable-Broadcast}
\newcommand{\SBCname}{f+1-Provable-Broadcast}

\newcommand{\SBC}[2]{\emph{f+1-PB}(#1,#2)}

\newcommand{\BC}[2]{4S-f+1-PB-broadcast(#1,#2)}

\newcommand{\lr}[1]{\langle #1 \rangle}

\newcommand{\ignore}[1]{}

\title{Validated Asynchronous Byzantine Agreement with Optimal
Resilience and Asymptotically Optimal Time and Word Communication} 
\author{Ittai Abraham\\
VMware Research\\
iabraham@vmware.com
\and Dahlia Malkhi\\
VMware Research\\
dmalkhi@vmware.com
\and Alexander Spiegelman\\
VMware Research\\
spiegelmans@vmware.com}

\date{}

\begin{document}
\maketitle

\begin{abstract}
We provide a new protocol for Validated Asynchronous Byzantine
Agreement. Validated (multi-valued) Asynchronous Byzantine Agreement is a key building block in constructing Atomic Broadcast and fault-tolerant state machine replication in the asynchronous setting.
Our protocol can withstand the optimal number $f<n/3$ of Byzantine failures and reaches agreement in the asymptotically optimal expected $O(1)$ running time.
Honest parties in our protocol send only an expected $O(n^2)$
messages where each message contains a value and a constant number of
signatures.
 Hence our total expected communication is $O(n^2)$ words.
The best previous result of Cachin et al. from 2001 solves Validated Byzantine Agreement with optimal resilience and $O(1)$ expected time but with $O(n^3)$ expected word communication.
Our work addresses an open question of Cachin et al. from 2001 and improves the expected word communication from $O(n^3)$ to the asymptotically optimal $O(n^2)$.

\end{abstract}

\section{Introduction}
\label{sec:intro}


Byzantine agreement is a fundamental problem in computer science introduced by Pease, Shostak and Lamport \cite{PSL80} in 1980. 
Bracha \cite{B87} shows that even strictly weaker primitives than Asynchronous Byzantine agreement can
only be solved when the number of parties $n$ is larger than $3f$ where $f$ is the maximum number of  parties the adversary can corrupt. We therefore say that a solution has \textit{optimal resilience} if it solves Byzantine agreement for $n=3f+1$.
A theorem of Fischer, Lynch and Paterson~\cite{FLP85} states that any
protocol solving Asynchronous Agreement must have a non-terminating execution even in the face of  a single (benign) failure. Ben-Or~\cite{Ben-Or} shows that randomization can be used to make such non-terminating executions become events with probability 0.
Feldman and Micali \cite{FM88} show that Asynchronous Byzantine Agreement can be
solved with optimal resilience $n=3f+1$ and with an expected $O(1)$
asynchronous running time (where running time is the maximum duration
as defined by Canetti and Rabin \cite{CR93} and is essentially the
number of steps when the protocol is embedded into a lock-step timing
model).
We therefore say that a solution has \textit{asymptotically optimal
time} if it solves Byzantine agreement using an expected $O(1)$
running time.
We show that a recent lower bound of Abraham et al. \cite{ADDNR18} implies that any protocol solving Asynchronous Byzantine Agreement aganst an adaptive adversary (and without a constant error probability) must have the honest parties send expected $\Omega(n^2)$ messages (see Appendix \ref{app:LB}). We therefore say that a solution has \textit{asymptotically optimal word communication} if it solves Byzantine agreement using an expected $O(n^2)$ messages and each message contains just a single \textit{word} where we assume a word contains a constant number of signatures and domain values. 
Renewed interest in Byzantine Agreement follows from the need to
implement Atomic Broadcast and fault tolerant state machine
replication in the asynchronous
setting~\cite{CV17,HoneyBadger,beat18}. In 2001, Cachin, Kursawe,
Petzold, and Shoup~\cite{CachinSecure} defined the problems of
Atomic Broadcast and Validated Byzantine Agreement to address these types of applications.  \emph{Validated Byzantine Agreement} guarantees a decision on some party's input satisfying a globally verifiable \emph{external validity} condition.  Cachin et al.~\cite{CachinSecure} show how to obtain Validated Asynchronous Byzantine Agreement (VABA) with optimal resilience, asymptotically optimal time and $O(n^3)$ expected word communication. Improving the expected word communication for VABA from $O(n^3)$ to the $O(n^2)$ is an open problem stated in \cite{CachinSecure} and has been open for almost 20 years.

This paper presents the first VABA solution with optimal resilience, asymptotically optimal time
whose expected word communication is $O(n^2)$, thus closing this gap.
%
More precisely, we prove the following theorem:

\begin{theorem}
There exists a protocol among $n$ parties that solves except with negligible probability Validated
Asynchronous Byzantine Agreement (VABA), secure against
an adaptive adversary that controls up to $f < n/3$ parties,
with expected $O(n^2)$ word communication and expected constant running time.
\end{theorem}

\paragraph{Background.}

In Byzantine Agreement there are $n$ parties each of which has an input value, at most $f < n$ are corrupted (i.e., controlled by an
adversary), and the goal of the honest parties is to decide on a
unique value.
Many models with various assumptions have
been proposed in the literature. 
Some consider asynchronous communication whereas others
rely on synchrony. 
Some restrict the adversary's computational power in order to use
cryptographic tools whereas others assume an information theoretic
model, and some assume authenticated or private communication channels
whereas others deal with possible forgeability.
In this paper, 
we assume the practical random oracle model~\cite{bellare1993random,
fiat1986prove, Cachin2000RandomOI}. 
Namely, an environment with authenticated but asynchronous
communication channels and a computationally bounded adversary.

A  simplified version of the agreement
problem is the \textit{binary agreement problem} in which the inputs of the parties are restricted to the set $\{0,1\}$.
A fundamental work by Cachin, Kursawe, and Shoup~\cite{Cachin2000RandomOI} was the
first to give an optimal algorithm in terms of resilience and
word communication in the random oracle model, which they
formalized to fit the distributed settings.
In particular, the algorithm withstands up to $f < n/3$ Byzantine
failures, runs in constant expected number of asynchronous views
(rounds), and the expected communication cost is $O(n^2)$ messages of
the size of one or two RSA signatures~\cite{shoup2000practical}.
A more recent work by Most{\'e}faoui et
al.~\cite{mostefaoui2015signature} shows how to achieve the same
optimal result without any cryptographic assumptions besides the
existence of a common random coin\footnote{while the construction of
\cite{mostefaoui2015signature} requires only $O(n^2)$ bits
given a common coin, the word
communication of the resulting binary Byzantine Agreement protocol is
dominated by the common random coin protocol that requires threshold
signatures and $O(n^2)$ word communication.}.

As for the multi-valued Byzantine agreement, the original 
problem specification due to Lamport et al.~\cite{PSL80}
was motivated by the
following setting: Four computers in control of a space-shuttle cockpit need to
reach agreement on a sensor reading, despite one being potentially faulty. The problem was captured via a \emph{Weak
Validity}~\cite{dolev1983authenticated} condition as follows:
\begin{definition}[Weak validity]

If all honest parties propose $v$, then every honest party that
terminates decides~$v$.

\end{definition}
\noindent Note that while the Weak Validity condition is well defined, it
says nothing about a situation in which parties propose different
values, allowing them to (1) return some
default value $\bot$ that indicates that no agreement was reached or
(2) agree on a value proposed by a corrupted party.
Most{\'e}faoui et al.~\cite{Raynalmultivalued} consider a slightly
stronger property in which only a value proposed by an honest party
or $\bot$ are allowed to be returned.
However, honest parties may still decide $\bot$ if they initially disagree.  In particular, it is not clear how this slightly stronger validity property can be used to solve Atomic Broadcast \cite{CachinSecure}.

Cachin et al. formulated in~\cite{CachinSecure} a problem specification that captures the practical settings where parties propose updates to a replicated state. Agreement is formed on a sequence of updates, hence a
non-default decision is needed in order to make progress.  To prevent updates from rogue parties, 
the model is extended with an \emph{External Validity} predicate 
as follows:
\begin{definition}[External validity]

If an honest party decides on a value $v$,
then $v$ is externally valid.

\end{definition}

Mostéfaoui et al. presented in~\cite{Raynalmultivalued} a
signature-free deterministic reduction
from their binary agreement protocol~\cite{mostefaoui2015signature} that solved asynchronous Byzantine Agreement with Weak validity. It has optimal resilience and asymptotically optimal time and word communication. However, the weak validity  property seems to prohibit the usefulness as a building block for Atomic Broadcast or any State Machine Replication (SMR) protocol that should maintain liveness in an asynchronous environment.

Cachin et al. gave in~\cite{CachinSecure} a randomized reduction
from their binary agreement algorithm~\cite{Cachin2000RandomOI}
to VABA and also showed how to use it in order to implement an atomic
broadcast.
Their VABA protocol provides external validity, has optimal
resilience, asymptotically optimal time, and expected message
complexity $O(n^3)$. 
That paper explicitly mentions the open problem of improving the
expected word communication from $O(n^3)$ to $O(n^2)$.

\paragraph{Our Contribution.}

The main contribution of this paper is solving this open question. 
Just like \cite{CachinSecure}, our protocol solves Asynchronous
Byzantine agreement with external validity (VABA), has optimal
resilience and asymptotically optimal time. Improving on
\cite{CachinSecure}, our expected word communication is also
asymptotically optimal.
In particular, honest parties send a total expected $O(n^2)$
messages, which is optimal and each message is roughly the size of one
or two threshold signatures.

Our protocol is secure against an adaptive adversary. This follows from using adaptively secure threshold signatures of Libert et al. \cite{LJY16} and adaptively secure common coin protocol of Loss and Moran \cite{Julian}. 
Cachin et al. \cite{CachinSecure} note that their binary protocol \cite{Cachin2000RandomOI} and their Validated protocol  \cite{CachinSecure} also immediately generalize to be secure against adaptive adversaries by using the primitives above.  


\paragraph{Techniques and Challenges.}
Our protocols are based on the Random Oracle model and draw heavily from the framework of Cachin et al. \cite{Cachin2000RandomOI,CachinSecure}.
Unlike previous constrictions, our protocol does not go through a randomized binary
agreement black-box.  
Instead, much like Katz and Koo's Synchronous Byzantine agreement protocol
\cite{KK09}, in each view, we run $n$ parallel leader-based threads
and then use a random leader election primitive to decide which
leader is elected in hindsight. Just like \cite{KK09}, not all honest
parties reach agreement in the same view. To guarantee safety between
different views we use a view-change protocol that guarantees that
new leaders can propose only safe values.

Adopting this approach to get optimal word communication requires to
overcome several challenges. On the one hand, to obtain $O(n^2)$ word
communication per view, we need each of the $n$ leader-based thread
protocol to use just $O(n)$ words and we need the global view change
protocol to use $O(n^2)$ words. On the other hand, to guarantee
progress we make sure that our view change protocol will allow progress even in asynchronous settings.
To balance between frugal communication and liveness we adopt a four step protocol that is inspired by an approach taken in the 
partial-synchrony model by Yin et al.~\cite{HS}.

To obtain optimal $O(1)$ expected time, the next challenge is to
guarantee that when the first honest party enters the leader election
phase there is a constant fraction of potential leaders such that if
one of them is elected then all honest parties will decide in
constant time.
Moreover, if the elected leader did not complete its broadcast, we
need a mechanism to allow parties to abandon the elected leader's
broadcast before running the global view-change protocol.

\section{Model}
\label{sec:model}

In order to reason about distributed algorithms in cryptographic
settings we adopt the model defined in~\cite{Cachin2000RandomOI,
CachinSecure}.
We consider an asynchronous message passing system consisting of a set
$\Pi$ of $n$ parties, an adaptive adversary, and a trusted dealer.
The adversary may control up to $f < n/3$ parties during an
execution.
An adaptive adversary is not restricted to choose which parties to
corrupt at the beginning of an execution, but is free to corrupt (up to $f$)
parties on the fly.
Note that once a party is corrupted, it remains corrupted, and we
call it \emph{faulty}.
A party that is never corrupted in an execution is called
\emph{honest}.
We denote the set $\Pi_h \subseteq \Pi$ to be the set of all honest
parties.

%
We assume an initial setup before every execution in which the trusted
dealer generates the initial states of all parties, and we assume
that the adversary cannot obtain the states of honest parties at
any time during an execution.
In particular, the adversary cannot obtain the initial states of
honest parties.

\paragraph{Computation.}
Following~\cite{Cachin2000RandomOI,CachinSecure}, we use
standard modern cryptographic assumptions and definitions. 
We model the computations made by all
system components as probabilistic Turing machines, and bound the
number of computational basic steps allowed by the adversary
by a polynomial in a \emph{security parameter k}.
A function $\epsilon(k)$ is \emph{negligible} in $k$ if for all $c >
0$ there exists a $k_0$ s.t.\ $\epsilon(k) < 1/k^c$ for all $k > k_0$.
A computational problem is called \emph{infeasible} if any polynomial
time probabilistic algorithm solves it only with negligible
probability.
Note that by the definition of infeasible problems, the probability to
solve at least one such problem out of a polynomial in $k$ number of
problems is negligible.
Intuitively, this means that for any protocol $P$ that uses a
polynomial in $k$ number of infeasible problems, if $P$ is correct
provided that the adversary does not solve one of its infeasible
problems, then the protocol is correct except with negligible
probability.
We assume that the number of parties $n$ is bounded by a polynomial
in $k$.

\paragraph{Communication.}
We assume asynchronous links controlled by the adversary,
that is, the adversary can see all messages and decide when and what
messages to deliver.
In order to fit the communication model with the computational
assumptions, we restrict the adversary to perform no more than a
polynomial in $k$ number of computation steps between the time a
message $m$ from an honest party $p_i$ is sent to an honest party
$p_j$ and the time $m$ is delivered by $p_j$\footnote{Note that
although this restriction gives some upper bound on the communication
in terms of the adversary local speed, the model is still asynchronous
since speeds of different parties are completely unrelated.}.
In addition, for simplicity, we assume that messages are
\emph{authenticated} in a sense that if an honest party $p_i$
receives a message $m$ indicating that $m$ was sent by an honest
party $p_j$, then $m$ was indeed generated by $p_j$ and sent to $p_i$
at some prior time.
This assumption is reasonable since it can be easily implemented with
standard symmetric-key cryptographic techniques~\cite{Bellare} in our
model.

\paragraph{Termination.}
Note that the traditional definition of the liveness property in
distributed system, which requires that all correct (honest) parties
\textit{eventually} terminate provided that all messages between
correct (honest) parties eventually arrive, does not make sense in
this model.
This is because the traditional definition allows the following:
\begin{itemize}
  
  \item Unbounded delivery time between honest parties, which
  potentially gives the adversary unbounded time to solve infeasible
  problems.
  
  \item Unbounded runs that potentially may consist of an unbounded
  number of infeasible problems, and thus the probability that the
  adversary menages to solve one is not
  negligible.
  
\end{itemize}

\noindent Follwong Cachin et al. \cite{Cachin2000RandomOI,CachinSecure},
we address the first concern by restricting the number of
computation steps the adversary makes during message transmission among
honest parties.
So as long as the total number of messages in the protocol is
polynomial in $k$, the error probability remains negligible.
To deal with the second concern, we do not use a standard liveness
property in this paper, but instead we reason about the total number
of messages required for all honest parties to terminate.
We adopt the following definition
from~\cite{Cachin2000RandomOI,CachinSecure}:

\begin{definition}[Uniformly Bounded Statistic]

Let $X$ be a random variable.
We say that $X$ is \textit{probabilistically uniformly bounded} if
there exist a fixed polynomial $T(k)$ and a fixed negligible functions
$\delta(l)$ and $\epsilon(k)$  such that for all $l,k \geq 0$, 
\[
Pr[X > lT(k)] \leq \delta(l) + \epsilon(k)
\]  

\end{definition}

\noindent With the above definition Cachin et
al.~\cite{Cachin2000RandomOI, CachinSecure} define a progress
property that makes sense in the cryptographic settings:
\begin{itemize}
  
  \item \emph{Efficiency: The number of messages generated by the
  honest parties is probabilistically uniformly bounded}
  
\end{itemize} 
\noindent The efficiency property implies that the probability of the
adversary to solve an infeasible problem is negligible, which makes
it possible to reason about the correctness of the primitives'
properties.
However, note that this property can be trivially satisfied by
a protocol that never terminates but also never sends any messages.
Therefore, in order for a primitive to be meaningful in
this model, Cachin et al.~\cite{Cachin2000RandomOI, CachinSecure}
require another property:
\begin{itemize}
  
  \item \emph{Termination\footnote{Called liveness
  in~\cite{Cachin2000RandomOI}, but we find this name confusing since
  it is not a liveness~\cite{alpern1985liveness} property.}}:
  If all messages sent by honest parties have been delivered, then
  all honest parties terminated.

\end{itemize}
In this paper we consider both efficiency and termination properties
as defined in~\cite{Cachin2000RandomOI, CachinSecure}.
However, note that when considering an adaptive adversary, it is also
possible to define a slightly weaker termination property:

\begin{itemize}
  
  \item \emph{Weak termination}:
  If all massages sent by parties before they were corrupted have
  been delivered, then all honest parties terminated.

\end{itemize}
%
%
Note that while any protocol that satisfies termination satisfies
weak termination as well, a lower bound for 
termination does not apply for weak termination.
Indeed our lower bound (see Appendix ~\ref{app:LB}) is for protocols that obtain the termination property.
We leave the study of lower bounds for protocols with weak termination as an open question.

\paragraph{Complexity.}

We use the following standard complexity notions (see for example Cannetti and Rabin \cite{CR93}).
We measure the expected \emph{word communication} of our protocol as the
maximum over all inputs and applicable adversaries of the expected 
total number of \textit{words} sent by honest parties where
expectation is taken over the random inputs of the players and of the
 adversary. 
We assume a finite domain $\mathcal{V}$ of valid values for
the Byzantine agreement problem, and say that a word can contain a
constant number of signatures (see Section \ref{sub:crypto}) and
domain values.
We measure the expected \emph{running time} of our protocol as
the maximum over all inputs and applicable adversaries
of the expected \textit{duration} where expectation is taken over the
random inputs of the players and of the adversary. The duration of an
execution is the total time until all honest players have terminated
divided by the longest delay of a message in this execution. 
Essentially the duration of an execution is the number of steps taken
if this execution is re-run in lock-step model where each message
takes exactly one time step.

%
%


%

Following Cachin et al. \cite{CachinSecure}, in order to show that our view-based protocol runs in an
expected constant running time and has expected $O(n^2)$ word
communication, it is enough to show that:
\begin{itemize}
  
  \item every view consists of $R(k) =O(n^2)$ messages that consist
  of one word, and
  
  \item  the total number of messages is probabilistically uniformly
  bounded by $R$.
  
\end{itemize}

(Note that the number of parties $n$ is a polynomial in $k$, and thus,
so is $R$). This follows from the following Lemma:

\begin{lemma}[CKPS 01~\cite{CachinSecure}]
Given a probabilistically uniformly bounded by $T$ random
variable $X$, there is a constant $c$ s.t.\ the expected value of $X$
is bounded by $cT(k)$ + $\epsilon(k)$, where $\epsilon(k)$ is a
negligible function.
\end{lemma}

\subsection{Cryptographic abstractions}
\label{sub:crypto}

The main focus of this paper is on a novel distributed algorithms. Our
protocol uses cryptographic tools as black-boxes. To this end, 
we present our protocol assuming the existence of the cryptographic
abstractions as defined bellow.

\paragraph{Threshold signatures scheme.}

At the beginning of every execution, every party $p_i$ gets a private
function $\emph{share-sign}_i(m)$ from the dealer, which gets a
message $m$ and returns a 
signature-share $\sigma_i$. 
In addition, every party gets the following functions: (1)
$\emph{share-validate}(m,i,\sigma_i)$, which gets a message $m$, a party
identification $i$, and a signature-share $\sigma_i$, and returns
\emph{true} or \emph{false}; (2)
$\emph{threshold-sign}(\Sigma)$, which gets a set of signature-shares
$\Sigma$, and returns a threshold signature $\sigma$; and (3)
$\emph{threshold-validate}(m,\sigma)$, which gets a message $m$ and a
threshold signature $\sigma$, and returns \emph{true} or \emph{false}.
We assume that the above functions satisfy the following properties
except with negligible probability:

\begin{itemize}
  
%
  \item \emph{Share validation:} For all $i$, $1 \leq i \leq n$ and
  for every messages $m$, (1) $\emph{share-validate}(m,i,\sigma) = true$ if and only if $\sigma =
  \emph{share-sign}_i(m)$, and (2) if $p_i$ is
  honest, then it is infeasible for the adversary to compute
  $\emph{share-sign}_i(m)$.
  

   \item \emph{Threshold validation:} For every message $m$,
  $\emph{threshold-validate}(m,\sigma) = true$ if and only if
  $\sigma = \emph{threshold-sign}(\Sigma)$ s.t.\
  $|\Sigma| \geq 2f+1$ and for every $\sigma_i \in \Sigma$ there is a
  party $p_i$ s.t.\ $\emph{share-validate}(m,i,\sigma) = true$.

  
\end{itemize}

\paragraph{Threshold coin-tossing.}

We assume an unpredictable pseudo random generator (PRG) $G: S \to
\{1,\ldots,n\}$, that is known only to the dealer, which gets a
string $s \in S$ and returns a party $p \in \Pi$.
Following the standard cryptographic definitions, unpredictability
means that if $G$ is not known, then given
evaluations of $G$ at all points in some set $Q$, 
the advantage in evaluating $G$ at a point not in $Q$
is negligible.
(See formal definition
in~\cite{Cachin2000RandomOI,CachinSecure,LJY16,Julian}).
At the beginning of every execution, the dealer gives a private
function $\emph{coin-share}_i(s)$ to every party $p_i$, which gets a
string $s$ and returns a coin share $\sigma_i$.
In addition, two public functions are available to all parties: (1)
$\emph{coin-share-validate}(s,i,\sigma_i)$, which gets a string $s$, a
party identification $i$, and a coin share $\sigma_i$, and returns
\emph{true} or \emph{false}; and (2) $\emph{coin-toss}(s,\Sigma)$,
which gets a string $s$ and a set of coin shares, and returns a party
in $\Pi$.

We assume that the following properties are satisfied   
except with negligible probability:

\begin{itemize}

  \item For all $i$, $1 \leq i \leq n$ and for every string $s$, (1)
  $\emph{coin-share-validate}(s,i,\sigma) = true$ if and only if
  $\sigma = \emph{coin-share}_i(s)$, and (2) if $p_i$ is
  honest, then it is infeasible for the adversary to compute
  $\emph{coin-share}_i(s)$.
  
   \item For every string $s$,
  $\emph{coin-toss}(s,\Sigma) = G(s)$ if and only if
  $|\Sigma| \geq f+1$ and for every $\sigma \in \Sigma$ there is a
  party $p_i$ s.t.\ $\emph{coin-share-validate}(s,i,\sigma) = true$.  
  
\end{itemize}

\paragraph{Implementations.}
Several widely used and established implementations of these
abstractions can be found in the literature,
e.g.,~\cite{Cachin2000RandomOI, CachinSecure, shoup2000practical},
and few more recent ones~\cite{LJY16,Julian} provide implementations
that are also proven to be secure against an adaptive adversary.
In all implementations, given an input value of size $B$, the size of
the returned shares (from the \emph{share-sign} and \emph{coin-share}
functions) and the threshold signature (the return value of
\emph{threshold-sign}) is $B + \epsilon(k)$, where $\epsilon(k)$ is a
negligible function.

\subsection{Validated  Asynchronous Byzantine Agreement (VABA)}

In this paper we follow Cachin et al.~\cite{CachinSecure} and
define a (multi-valued) Byzantine agreement with 
an external function we call
\emph{ex-ba-validation}, that determines whether a value is valid for
agreement or not.
In addition, we explicitly require a notion of \emph{quality} to capture
the ``fairness'' of the decision value. Note that the protocol of
Cachin et al.
\cite{Cachin2000RandomOI,CachinSecure} already obtains this Quality property.


\begin{definition}[Validated Byzantine Agreement]

A protocol solves validated   Byzantine agreement with chain quality if it satisfies the
following properties except with negligible probability:

\begin{itemize}
  
  \item Validity: If an honest party decides an a value $v$, then
  $\emph{ex-ba-validation}(v) = true$.
  
  \item Agreement: All honest parties that terminates decide on the
  same value.


  \item Quality: The probability of choosing a value that was proposed by an
  honest party is at least $1/2$.
  

  \item Termination\footnote{Called liveness
  in~\cite{Cachin2000RandomOI}, but we find this name confusing since
  it is not a liveness~\cite{alpern1985liveness} property.}:
  If all honest parties start with externally valid values and all
  massages sent among honest parties have been delivered, then all
  honest parties decided.
  
  \item Efficiency: The number of messages generated by the
  honest parties is probabilistically uniformly bounded.

\end{itemize}

\end{definition}


\section{VABA with Optimal Resilience and Asymptotically Optimal time and Communication}
\label{sec:protocol}

In this section we give a protocol for asynchronous Byzantine agreement, secure against an adaptive adversary that controls up to $f < n/3$ parties,
with expected word communication $O(n^2)$ and expected running time $O(1)$.
Inspired by Cachin et al.~\cite{CachinSecure, CachinOPODIS,
Cachin2000RandomOI} we present a modular implementation of our
protocol, which consists of three sub protocols: two broadcast
primitives we call \emph{\BCname} and \emph{\SBCname}, and a simple
and efficient leader election protocol.

\subsection{Broadcast primitives}

A broadcast primitive is an abstraction for a pre-defined party,
which is called a \emph{sender}, to pass a message to all other
parties. An instance of a broadcast primitive is identified via an \emph{id}.
We present two broadcast abstractions: \emph{\SBCname} and \emph{\BCname}, the
first is a simple but useful primitive that we use
to implement the second, which in turn is used for the Byzantine agreement protocol.
In both, a sender broadcasts a message $m = \lr{v,\sigma}$ consisting
of a value $v$ and a proof $\sigma$. The proof is provided in order to
allow recipients to screen messages.
Both abstractions are parametrized with an \textit{external validation}
function that parties use for screening messages.
These functions implement an important logic that drive the safety properties
of our Byzantine agreement protocol.

The API of our broadcasts somewhat differs from other broadcasts in the
literature in two ways:
\begin{itemize}
  
  \item The sender's broadcast returns a proof $\sigma$.
  Informally, $\sigma$ proves, via the threshold
  signature abstraction, that at least $f+1$ honest parties delivered
  $m.v$.
  
  \item Parties can invoke \emph{abandon(id)} to explicitly stop their
participation in the broadcast protocol.
  %
\end{itemize}

\subsubsection{\SBCname}

An \SBCname\ of a message $m$ with identification $id$ is denoted
\SBC{$id$}{$m$}.
The external validation function of a party $p_i$ is denoted 
\emph{ex-sbc-validation}$_i(id, m)$.
\SBC{$id$}{$m$} satisfies the following properties
except with negligible probability:


\begin{itemize}
  
  \item \textbf{PB-Integrity:} An honest party delivers a message at
  most once.
  
  \item \textbf{PB-Validity:} If an honest party $p_i$ delivers $m$,
  then \emph{ex-sbc-validation}$_i(id,m) = true$.
  
  \item \textbf{PB-Abandon-ability:} An honest party does not deliver
  any message after it invokes abandon(id). 
  
  \item \textbf{PB-Provability:} For all $v,v'$, if the sender can
  produce strings $\sigma,\sigma'$ s.t.\\
  $\emph{threshold-validate}(\lr{id,v},\sigma) = \emph{true}$ and
  $\emph{threshold-validate}(\lr{id,v'},\sigma') = \emph{true}$, then
  (1) $v = v'$ and (2) $f+1$ honest parties delivered a
  message $m$ s.t.\ $m.v = v$.
  
  \item \textbf{PB-Termination:} If the sender is honest, no
  honest party invokes abandon(id), all messages among honest parties
  arrive, and the message $m$ that is being broadcast is
  externally valid, then (1) all honest parties deliver $m$, and (2)
  \SBC{id}{m} returns (to the sender) $\sigma$, which satisfies\\
  $\emph{threshold-validate}(\lr{id,m.v},\sigma) = \emph{true}$.

  \item \textbf{PB-Linear-complexity:} The total number of messages sent by
  honest players is at most $2n$. 
  
\end{itemize}


The pseudocode of \SBCname\ appears in
Algorithms~\ref{alg:SBCsender} and~\ref{alg:SBCall}.
The sender sends a message $m = \lr{v,\sigma}$ consisting of a value
$v$ and a proof $\sigma$ to all parties.
Once a party gets a message $m$ from the sender that passes its
ex-sbc-validation function for the first time it produces a valid
signature share of $\nu$ on $\lr{id,m.v}$ , delivers the message $m$,
and sends the share $\nu$ back to the sender.
When the sender receives $2f+1$ valid shares it produces a valid
threshold signature and returns it.
If \emph{abandon(id)} is invoked by a party, it ignores all
further messages of this broadcast instance.

\begin{algorithm}
\caption{\SBCname\ with identification $id$: Protocol for the sender} 
\begin{algorithmic}[1]

\Statex \textbf{Local variables initialization:}
\State $S = \{\}$

\Statex

\Upon{\SBC{$id$}{$\lr{v,\sigma}$} invocation}

\State send ``$id,\emph{send},\lr{v,\sigma}$''
to all parties \State \textbf{wait} until $|S| = 2f+1$
\State  \textbf{return} $\emph{threshold-sign}(S)$  
\Statex

\EndUpon

\Receiving{``$id,\emph{ack},\nu_j$'' form party $p_j$ for the first
time}

\If{$\emph{share-validate}(\lr{id,v},j,\nu_j) = \emph{true}$}

	\State $S \gets S \cup \{\nu_j\}$

\EndIf

\EndReceiving

\end{algorithmic}
\label{alg:SBCsender}
\end{algorithm}

\begin{algorithm}
\caption{\SBCname\ with identification $id$: Protocol for a party
$p_i$} \begin{algorithmic}[1]

\Statex \textbf{Local variables initialization:}
\State $\emph{stop} \gets false$

\Statex

\Receiving{``$id,\emph{send},\lr{v,\sigma}$''
from the sender }
	
	\If{$\emph{stop} = false \wedge
	\emph{ex-sbc-validation(id,$\lr{v,\sigma}$)} =
	\emph{true}$}
		
		\State $\emph{stop} \gets true$
		\State $\nu_i \gets \emph{share-sign}_i(\lr{id,v})$
		\State \textbf{deliver} $\lr{v,\sigma}$
		\State send ``$id,\emph{ack},\nu_i$'' to the sender
	
	\EndIf

\EndReceiving

\Statex

\Upon{abandon($id$)}
 
	\State $\emph{stop} \gets true$

\EndUpon

\end{algorithmic}
\label{alg:SBCall}
\end{algorithm}

\subsubsection{\BCname}

A \BCname\ of message $m$ with identification $id$ is denoted by
\BC{$id$}{$m$}.
The external validation function of a party $p_i$ is denoted 
\emph{ex-bc-validation}$_i(id,m)$. \BCname\ has three deliveries
which we refer to as \textbf{key}, \textbf{lock}, and
\textbf{commit},
respectively.
These deliveries satisfy similar properties to \SBCname, but where 
\textbf{PB-provability} additionally guaranteeing that a \textbf{lock} delivery implies that
a \textbf{key} delivery has occurred in at least $f+1$ honest parties, and a
\textbf{commit} delivery implies that a \textbf{lock} delivery has occurred in at
least $f+1$ honest parties.

\BCname\ is implemented on top of \SBCname, and consists of 
four sequential invocations of \SBCname\ with
identifications $\lr{id,j}$ for $j \in \{1,\ldots,4\}$.
\SBC{$\lr{id,j}$}{$\lr{v,\lr{\sigma_{ex},\sigma_{in}}}$} is invoked
with two proofs: a proof $\sigma_{ex}$ for external validity
condition of ex-bc-validation that implements a logic of the VABA
protocol, and a proof $\sigma_{in}$ for external validity condition
of ex-sbc-validation that implements a logic of the \BCname.
\SBC{$\lr{id,1}$}{$\lr{v,\lr{\sigma_{ex},\sigma_{in}}}$} takes
$\sigma_{in} = \bot$, while the remaining invocations require
$\sigma_{in}$ to be a valid output of the previous one.

When a party delivers a message $\lr{v,\lr{\sigma_{ex},\sigma_{in}}}$
in the $j^{th}$ broadcast instance (the f+1-Provable- Broadcast with identification $\lr{id, j}$), it
delivers \textbf{key}$(id,\lr{v,\,\sigma_{in}})$,
\textbf{lock}$(id,\lr{v,\,\sigma_{in}})$, or
\textbf{commit}$(id,\lr{v,\,\sigma_{in}})$, respectively.
When an \emph{abandon(id)} is invoked by a party, it simply invokes
\emph{abandon($\lr{id,j})$} for all $j \in \{1,\ldots,4\}$.
We do not define and prove the properties of the \BCname. 
Instead, we use it to describe our validated asynchronous
Byzantine agreement protocol, which we prove with the \SBCname\
properties directly.

The pseudocode of \BCname\ appears in
Algorithms~\ref{alg:BCsender} and~\ref{alg:BCall}.

\begin{algorithm}
\caption{\BCname\ with identification $id$: Protocol for a sender.} \begin{algorithmic}[1]

  \Upon{\BC{$id$}{$\lr{v,\sigma_{ex}}$} invocation}
  
  	\State $\sigma_{in} \gets \bot$
  	
  	\For{$j =1,..,4$}
  	
  		\State $\sigma_{in} \gets \SBC{\langle id,j \rangle}{\langle 
  		v, \lr{\sigma_{ex},\sigma_{in}} \rangle}$
  	
  	\EndFor
  	\State \textbf{return} $\sigma$
  	
  \EndUpon
  
  \Statex

  \Statex \textbf{External validation function the for \SBCname
  instances:}
  
  \Procedure{ex-sbc-validation}{$\langle id,j \rangle, \langle 
  		v, \lr{\sigma_{ex},\sigma_{in}} \rangle$}

 	\If{$j=1 \wedge \emph{ex-bc-validation}(id,\lr{v,\sigma_{ex}}) = \emph{true}$}

 		\State \textbf{return} true
 		
 	\EndIf

 	\If{$j > 1 \wedge \emph{threshold-validate}(\lr{\lr{id,j-1},
 	v},\sigma_{in})$ = \emph{true}}
 	
 		\State \textbf{return} true
 	
 	\EndIf
 	
 	\State \textbf{return} false

%
%
%
%
%
%
%
%
%
%
%
%
 	
\EndProcedure

\end{algorithmic}
\label{alg:BCsender}
\end{algorithm}

\begin{algorithm}
\caption{\BCname\ with identification $id$: Protocol for all parties.} 
\begin{algorithmic}[1]

\Upon{delivery($\langle id,j \rangle ,
\lr{v,\lr{\sigma_{ex},\sigma_{in}}} $)}
\Comment{$j \in {1,\ldots,4}$}

	\If{$j = 2$}
		\State \textbf{deliver key}$(id,\lr{v,\sigma_{in}})$
	\EndIf
	
	\If{$j = 3$}
		\State \textbf{deliver lock}$(id,\lr{v,\sigma_{in}})$
	\EndIf
	
	\If{$j = 4$}
		\State \textbf{deliver commit}$(id,\lr{v,\sigma_{in}})$
	\EndIf

  \EndUpon
  
  \Statex

\Upon{abandon($id$)}

 	\For{$j=1,\ldots,4$}
 	
 		\State \emph{abandon($\langle id, j \rangle$)} 
 	
 	\EndFor
 
\EndUpon

\end{algorithmic}
\label{alg:BCall}
\end{algorithm}

\subsection{Leader election}

A leader election abstraction provides one operation to elect a unique
party (called a leader) among the parties.
An instance of a leader election primitive is identified via an \emph{id}, and exposes an
operation $elect(id)$ to all parties, which returns a party $p \in
\Pi$. Formal definitions are given below and the pseoudocode appears
in Algorithm~\ref{alg:election}.

A protocol for leader election associated with id $id$ satisfies
the following properties except with negligible probability.

\begin{itemize}
  
  \item \textbf{Termination:} If $f+1$ honest parties invoke
  $elect()$, and all messages among honest parties
  arrive, then all invocations by honest parties return.
  
  \item \textbf{Agreement:} All invocations of $elect(id)$ by honest
  parties return the same party.
  
  \item \textbf{Validity:} If an invocation of $elect(id)$ by an
  honest party returns, it returns a party $p$ with probability $1 /
  |\Pi|$ for every $p \in \Pi$.
  
  \item \textbf{Unpredictability:} The probability of the adversary
  to predict the returned value of $elect(id)$ invocation by an honest
  party before it returns is at most $1 / |\Pi|$ + $\epsilon(k)$,
  where $\epsilon(k)$ is a negligible function. 
  
\end{itemize}

\begin{algorithm}
\caption{Leader election. Protocol for party $p_i$} 
\begin{algorithmic}[1]

\Statex \textbf{Local variables initialization:}
\State $\Sigma \gets \{\}$

\Statex

\Upon{$elect(id)$}

	\State $\rho_i \gets \emph{coin-share}_i(id)$
	\State send ``$\textsc{share},id,\rho_i$'' to all parties
	\State \textbf{wait} until $|\Sigma| = f+1$
	\State \textbf{return} $\emph{coin-toss}(id,\Sigma)$

\EndUpon

\Statex

 \Receiving{``$\textsc{share},id,\rho_j$'' from party $p_j$ for
 the first time}
 
 	\If{$\emph{coin-share-validate}(id,j,\rho_j) = true$}
 	
 		\State $\Sigma \gets \Sigma \cup \{\rho_j\}$
 	
 	\EndIf
 
 \EndReceiving

\end{algorithmic}
\label{alg:election}
\end{algorithm}

\subsection{VABA protocol}
\label{sub:PABA}

\subsubsection{Overview.} 

The \BCname\ and the
leader election abstractions are used as building blocks for our
validated asynchronous Byzantine agreement protocol.
The protocol works in a view-by-view manner, where each view consists
of three phases: \emph{Broadcast Phase}, \emph{Leader-election Phase}, and
\emph{View-change Phase}.
In the \emph{Broadcast Phase}, each party invokes \BCname\
to broadcast its value. Parties wait to learn that $2f+1$ instances of
\BCname\ have \emph{completed} by returning at their respective senders.
Then, in the \emph{leader-election Phase}, parties choose a 
leader $p_l$ uniformly at random.
Finally, in the \emph{View-change Phase}, parties learn what
happened in the elected leader's \BCname.
If they learn that some party delivered \textbf{commit}(v) (and has a
proof that justifies it), they can decide $v$.
Otherwise, they need to carefully adopt a value and continue to the
next view. (more details are in the view-change description in
Section~\ref{subsub:PFRJ}). 

Our protocol guarantees that at least $2f+1$ \BCname\ instances complete in the
broadcast phase before a leader is chosen. 
If the elected leader has completed (implying that at least $f+1$
honest parties delivered a \textbf{commit}) then even an adaptive
adversary cannot prevent progress.
Otherwise, the view-change phase ensures that agreement is not
violated even if a bad leader is chosen.
Since the probability to choose a leader that completed its
broadcast is constant, the number of views in the protocol is
constant in expectation.
More concretely, the probability to choose a completed broadcast is
greater than 2/3, and thus the number of views in expectation is less
than $3/2$.
More details on each phase are given bellow.
The pseudocode appears in Algorithms~\ref{alg:BAmessages}
and~\ref{alg:BAoperation}, and a formal proof is given in
Appendix~\ref{app:proofs}

\subsubsection{Protocol for view $j$}
\label{subsub:PFRJ}

\paragraph{Broadcast phase.}
Each party broadcasts, using a \BCname, the value and \emph{KEY} it has
adopted from the previous views, as determined in the View-Change
Phase (explained below); at view $1$, a party broadcasts its input and an empty key. 
Parties participate in $n$ concurrent \BCname s, where
their ex-bc-validation function (Algorithm~\ref{alg:BAoperation}
lines~\ref{line:BAvalidty_starts}-\ref{line:BAvalidty_ends}) uses
the external validity condition for agreement (i.e.,
ex-ba-validation function), as well as a condition on the \emph{KEY}
as explained below.

Each party sends a notification about the completion of its own \BCname\
carrying its output for proof. When a party receives $2f+1$ such notifications, it
sends a signature share on a ``skip'' message. Each party waits to obtain
(either directly or indirectly) a combined threshold skip signature, 
forwards it to others, and moves to the Leader-election phase (even
if its broadcast has not completed).

\paragraph{Leader election phase.}
Once a party enters the Leader-election Phase, it abandons all the
\BCname\ instances.
The parties elect a leader via the leader election abstraction
and continue to the next phase as if only the the leader's \BCname\
ever occurred.

From now on, we refer to the \BCname\ of the leader of a view $j$ 
as the \BCname\ of view $j$, and refers to its delivery events as the
deliveries of view $j$. 

\paragraph{View-change phase.}

In the View-change Phase of view $j$, parties report their deliveries
from the broadcast of this view (the leader’s broadcast)
to everyone, including a proof of delivery
for each report.
Each party waits to collect $2f+1$ reports.
Recall that \BCname\ provides the following guarantees: 
If some honest party delivered a \textbf{commit}, then $f+1$ honest parties delivered
\textbf{lock}, hence all honest parties collect this \textbf{lock} in the
view-change exchange.
Similarly, if some honest party delivered a \textbf{lock}, then $f+1$ honest
parties delivered \textbf{key}, hence all honest parties collect this
\textbf{key} in the view-change exchange.

All parties maintain two cross-view variables, \emph{LOCK} and
\emph{KEY}: 
\begin{itemize}
    \item The \emph{LOCK} variable stores the highest view
$\mathcal{R}$ for which the party ever received a view-change message
that includes a \textbf{lock}.

    \item The \emph{KEY} variable stores the highest view
$\mathcal{R}$, for which the party ever received a view-change message
that includes a \textbf{key} and the key itself. 
\end{itemize}

Once a party has collected $2f+1$ view-change messages, it processes them as
follows. 
If it receives a \textbf{commit} $v$, it decides $v$.
Otherwise, if it receives a \textbf{lock}, 
it increases its \emph{LOCK} variable to the current view.
Last, if it receives a \textbf{key}, it updates its \emph{KEY}
variable to store the current view and the received \textbf{key}.
If it did not reach a decision, a party 
adopts the value $v$ of its
(up-to-date) \emph{KEY} variable and moves to the next view, where
it broadcasts $v$ together with $\textit{KEY}$ as proof for the
external validation function (ex-bc-validation).

As mentioned above, a party participates in
a \BCname\ only if the message $m =\lr{v,
\lr{\mathcal{R},\textbf{key}}}$ (note
that $\lr{\mathcal{R},\textbf{key}} = \textit{KEY}$) passes its
external validation test.
The external validation includes a crucial \emph{key-locking
mechanism} (see Algorithm~\ref{alg:BAoperation},
lines~\ref{line:BAvalidty_starts} to~\ref{line:BAvalidty_ends}).
In particular, in view $j > 1$, a party checks that the 
\textbf{key} is valid for $v$ and $\mathcal{R}$ (meaning that
\textbf{key} includes a proof that \textbf{key}(v) could have
been delivered by an honest party in the chosen broadcast of view
$\mathcal{R}$), and that the view $\mathcal{R}$ is at least as large
as \emph{LOCK}.
We prove in Appendix~\ref{app:proofs} that the key-locking mechanism
together with the fact that parties abandon all broadcasts before
sending the view-change messages guarantee agreement and satisfy
progress.
Here we give some intuition for the proof:

\begin{itemize}
  
  \item Lock Safety: If some party has a proof for \textbf{commit} $v$ in view
  $\mathcal{R}$, then at least $f+1$ honest parties previously
  locked (\emph{lock} = $\mathcal{R}$) in view $\mathcal{R}$.
  
  \item Key Safety:  If some party has a proof for \textbf{commit}
  $v$ in view $\mathcal{R}$, then it is not possible for a party to
  have a valid \textbf{key} on a value other than $v$ in view higher
  than or equal to $\mathcal{R}$.
  
  \item Key Progress: If some party $p_i$ is locked
  in view $\mathcal{R}$, then at least $f+1$ honest parties
  obtained a \textbf{key} in $\mathcal{R}$ before sending the
  view-change messages of view $\mathcal{R}$, and thus all honest
  parties will have a \emph{KEY} with view at least
  $\mathcal{R}$.
  
\end{itemize}

\paragraph{Communication complexity.}
We start by analyzing the word communication of a single view.
Recall that a word can contain a constant number of domain values
and signatures. We denote by  $V$ the domain of valid values for
the Byzantine agreement.
Note that every message that is broadcasted during the agreement
protocol consists of a value $v \in V$ and two
proofs - a proof (key) for the ex-bc-validation function of the
\BCname\ and another proof for the ex-sbc-validation function of every
instance of the \SBCname s therein.
Since a proof is simply a threshold signature on a value and the
broadcast identification, we get that the size of the
broadcast messages is a single word.
The total number of messages sent by honest parties in an
\SBCname\ is $2n$.
Therefore, since \BCname\ uses 4 instances of \SBCname, we get
that the number of words sent by honest parties in the \SBCname\
protocol is $O(n)$.
%
%

Our Byzantine agreement protocol uses $n$ concurrent \BCname s in
phase 1 of every view, which brings the word communication to
$O(n^2)$. 
The word complexity of sending the all to all ``skip-share'' and
``skip'' messages as well as the message complexity of the leader
election abstraction is $O(n^2)$.
Finally, since each party sends $n$ view change messages, of a
single word (at most 3 values and 3 proofs), we get that the word
communication of the view-change phase is $O(n^2)$.
In total, the word communication of a single view in our
Byzantine agreement protocol is $O(n^2)$.
We prove in Appendix~\ref{app:proofs} that our protocol has an
expected constant number of views, which implies that the expected
word communication of our asynchronous Byzantine agreement protocol
is $O(n^2)$.

\begin{algorithm}
\caption{Validated asynchronous byzantine agreement with
identification $id$:
protocol for party $p_i$.} \begin{algorithmic}[1]

\Statex \textbf{Local variables initialization:}

\Statex $LOCK \gets 0$
\Statex $KEY \gets \langle 0, v_i, \bot \rangle$ with selectors
\emph{view},
\emph{value}, \emph{proof}
\Statex \textbf{for every} view $j \geq 1$, \textbf{initialize}:
\Statex \hspace*{0.5cm}$v_j \gets v_i$ ; $\sigma_j \gets 
\lr{0, \bot}$ ;  $L[j] \gets \bot$; $BCdone_j \gets 0$; $BCskip_j
\gets \{\}$; $skip_j \gets false$
\Statex \hspace*{0.4cm} \textbf{for every} party $p_k \in \Pi$
\textbf{initialize}:
\Statex \hspace*{1cm} $D_{key}[k,j] =
D_{lock}[k,j] = D_{commit}[k,j] = \langle \bot, \bot \rangle$

%
%
%

\Statex

\Statex \textbf{External validity for the \BCname:}

\Procedure{ex-bc-validation}{$id,\langle v , \lr{j,\sigma}
\rangle$}
\label{line:BAvalidty_starts}

\If{\emph{ex-ba-validation}($v$) = false}
\label{line:ex-ba}
\Comment{external ABA validity check}

	\State \textbf{return} false

\EndIf

\If{$ j \neq 1 \wedge
\emph{threshold-validate}(\lr{\lr{\lr{id,L[j],j},1}, v}, \sigma$) =
false} \Comment{validate the key}
	
	\State \textbf{return} false

\EndIf

\If{$j \geq LOCK$}  \Comment{check that key is not smaller than lock}

	\State \textbf{return} true 
	
\Else

	\State \textbf{return} false

\EndIf

\EndProcedure
\label{line:BAvalidty_ends}

\Statex
\Statex \textbf{Protocol for party $p_i$}

\State $j \gets 1$
\While {true}

\ForAll{k=1,\ldots,n}
 \Comment{Broadcast phase}

	\State initialize an instance of \BCname\ with identification
	$\langle id,k,j \rangle$

\EndFor
\State $\sigma \gets$ \BC{$\langle id, i ,j \rangle$}{$\lr{v_j,
\sigma_j}$}

\State \textbf{wait} for  \BC{$\langle id, i ,j \rangle$}{$\lr{v_j,
\sigma_j}$} to
return or $skip_j = true$

\If{$skip_j = \emph{false}$}

	\State send ``$id, done,j,\lr{\lr{v_j, \sigma_j},\sigma}$'' to all
	parties

\EndIf

\State \textbf{wait} until $skip_j =true$

\Statex \Comment{Leader election phase}

\ForAll{k=1,\ldots,n}

	\State $abandon(\langle id, k ,j \rangle)$

\EndFor

\State $L[j] \gets $ \emph{$elect(\lr{id,j})$}

\Statex \Comment{View-change phase}

\State send
``$\emph{View-change},id,j,D_{key}[L[j],j],D_{lock}[L[j],j],D_{commit}[L[j],j]$''
to all parties

\State \textbf{wait for} \emph{View-change} messages from $2f+1$
different parties

\State $v_{j+1} \gets  KEY.value $
\State $\sigma_{j+1} \gets\lr{KEY.round, KEY.proof}$

\State $j \gets j + 1$
\label{line:endroundj}

\EndWhile

\end{algorithmic}
\label{alg:BAoperation}
\end{algorithm}

\begin{algorithm}
\caption{Validated asynchronous byzantine agreement with
identification $id$:
messages.} \begin{algorithmic}[1]

%
%
%

\Upon{$key(\langle id,k,j \rangle,\lr{v,\sigma})$}

	\State $D_{key}[k,j] = \lr{v,\sigma} $

\EndUpon

\Statex

\Upon{$lock(\langle id,k,j \rangle,\lr{v,\sigma})$}

	\State $D_{lock}[k,j] = \lr{v,\sigma} $

\EndUpon

\Statex

\Upon{$commit(\langle id,k,j \rangle,\lr{v,\sigma})$}

	\State $D_{commit}[k,j] = \lr{v,\sigma} $

\EndUpon

\Statex

 \Receiving{``$id,done,j,\lr{v,\sigma}$'' from party $p_i$ for
 the first time}

	\If{$ \emph{threshold-validate}(\lr{\lr{\lr{id,i,j},4},v},
	\sigma)$}

		\State $BCdone_j \gets BCdone_j + 1$

		\If{$BCdone_j = 2f+1$ and ``\emph{skip-share}'' message has not
		sent yet}

			\State $\nu \gets \emph{sigh-share}(\lr{id,skip,j})$
			\State send ``$id,\emph{skip-share},j,\nu$'' to all parties

		\EndIf

	\EndIf

 \EndReceiving

  \Statex

\Receiving{``$id,\emph{skip-share},j,\nu$'' from party $p_i$ for
 the first time}

   \If{\emph{share-validate}($\lr{id,skip,j},i,\nu)$}

		\State $BCskip_j \gets BCskip_j \cup \{\nu\}$

		\If{$|BCskip_j| = 2f+1$}

			\State $\sigma \gets \emph{threshold-sign}(BCskip_j)$
			\State send ``$id,skip,j,\sigma$'' to all parties

		\EndIf

   \EndIf

\EndReceiving

	 \Statex

\Receiving{``$id,skip,j,\sigma$''}

	\If{\emph{threshold-validate}($\lr{id,skip,j}, \sigma$) =
	true}

		\State $skip_j \gets true$
	\EndIf

	\If{``skip'' message was not sent yet}
		\State send ``$id,skip,j,\sigma$'' to all parties

 	\EndIf

 \EndReceiving
	
%
%
%
%
%
%
%
%
%
%
%
%
%
%
%
%
%
%
%
%
%

%
%
%
%
%
 
 \Statex

\Receiving{``$\emph{View-change},j,\langle v_2, \sigma_2
\rangle,\langle v_3, \sigma_3 \rangle, \langle v_4, \sigma_4
\rangle$''}

	\State $l \gets L[j]$

 	\If{$v_4 \neq \bot \wedge 
	\emph{threshold-validate}(\lr{\lr{\lr{id,l,j},3},v_4}, \sigma_4) =
	true$}
	
	\State \textbf{decide} $v_4$
	
	\EndIf

	\If{$v_3 \neq \bot \wedge j > lock \wedge
	\emph{threshold-validate}(\lr{\lr{\lr{id,l,j},2},v_3}, \sigma_3) =
	true$}
	
		\State $LOCK \gets j$

	\EndIf

	\If{$v_2 \neq \bot \wedge j > key.round \wedge
	\emph{threshold-validate}(\lr{\lr{\lr{id,l,j},1},v_2}, \sigma_2) =
	true$}
	
		\State $KEY \gets \lr{j,v_2,\sigma_2}$
	
	\EndIf
 
\EndReceiving

\Statex


\end{algorithmic}
\label{alg:BAmessages}
\end{algorithm}



\clearpage
\section{Discussion}
\label{sec:discussion}

Our protocol addresses an open problem of Cachin et al.
\cite{CachinSecure} and reduces the expected word communication from
$O(n^3)$ to $O(n^2)$ against an asynchronous adaptive adversary. We
also show that in the standard definition of an asynchronous adaptive
adversary this expected word communication is asymptotically optimal
for any protocol that obtains the standard definition of termination
(liveness) as defined~\cite{Cachin2000RandomOI,CachinSecure}.
An interesting open question is related to protocols that obtain weak termination in the adaptive setting:
is there a  $\Omega(n^2)$ lower bound
against an adaptive adversary that is required to deliver
all messages sent by parties before they are corrupted? or does there
exist a protocol with near linear expected word communication under this weak termination property?

\newpage
\begin{appendices}

\section{Correctness proofs}
\label{app:proofs}

 \subsection{\SBCname}
 
 Note that \emph{PB-integrity}, \emph{PB-validity},
 \emph{PB-abandon-ability}, and \emph{PB-linear-complexity} follow
 immediately from the code. 
 We now prove \emph{PB-provability} and \emph{PB-termination}.
 
\begin{lemma}

Algorithms~\ref{alg:SBCsender} and~\ref{alg:SBCall} satisfy
PB-termination. 

\end{lemma}

\begin{proof}

Consider a \SBCname\ instance with identification $id$ in which an
honest sender $p_i$ broadcasts an externally valid message $m$, no
honest party invokes \emph{abandon(id)}, and all messages among honest
parties arrived.
Since the sender is honest, it sent a message to all honest parties,
and since all messages among honest parties arrived, we get that all
honest parties received a message from the sender and thus all honest
parties delivered a message and (1) is satisfied.
For (2), note that the sender accept an ack message only if it
contains a valid share.
Since there are at least $2f+1$ honest parties, the sender gets at
least $2f+1$ valid shares, and thus by the threshold validation
property, it produces a valid threshold signature except with
negligble probability.

\end{proof}

\begin{lemma}

Algorithms~\ref{alg:SBCsender} and~\ref{alg:SBCall} satisfy
PB-provability. 

\end{lemma}
 
Consider an \SBCname\ instance with identification $id$ and assume
that a sender can produce two strings $\sigma_1,\sigma_2$ s.t.\
$\emph{threshold-validate}(\lr{id,v_1},\sigma_1) = 
\emph{threshold-validate}(\lr{id,v_2},\sigma_2) = \emph{true}$ for
some values $v_1,v_2$. 
By the threshold validation property, except with
negligible probability, for $j \in \{1,2\}$, $\sigma_j =
\emph{threshold-sign}(\Sigma_j)$ s.t.\ $|\Sigma_j| \geq 2f+1$ and for
every $\sigma_j' \in \Sigma_j$ there is a party $p_i$ s.t.\
$\emph{share-validate}(\lr{id,v_j},i,\sigma_1') = true$.
For $j \in \{1,2\}$, let $P_j = \{p_i \mid \exists \sigma_j' \in
\Sigma_j$ s.t.\ $\emph{share-validate}(\lr{id,v_j},i,\sigma_1') = true
\}$, and note that $|P_j| \geq 2f+1$.
Therefore, there are at lest $f+1$ honest parties in $P_j$, $j \in
\{1,2\}$, and thus there is an honest party $p_i$ that is in both
$P_1$ and $P_2$.
By the code, $p_i$ computes a \emph{share-sign(v)} at most once.
Therefore, by the share validation property we get that $v_1 = v_2 =
v$ except with negligible probability, which proves (1).
For (2), note that an honest party deliver a message $m$ before
computing a share on $m.v$.

\subsection{Validated asynchronous byzantine agreement}

\paragraph{Notations.}
Consider a byzantine agreement instance with identification $id$.
All parties start at view $1$, and we say that a party
\emph{completes} view $j \geq 1$ and \emph{moves} to view $j + 1$
when it executes line~\ref{line:endroundj} in
Algorithm~\ref{alg:BAoperation} for the $j^{th}$ time.
We say that view $j$ \emph{completes} when at least $f+1$ honest
parties move to view $j+1$.
The \emph{leader} of view $j$ is the party returned by the
$elect(\lr{id,j})$ invocations by honest parties at view $j$, and
we say that it is \emph{elected} when the first
such invocation returns.
We say that a \textbf{key} $= \lr{j,\sigma}$ is valid
for a value $v$ if
$\emph{threshold-validate}(\lr{\lr{\lr{id,l,j},1},v}, \sigma) =
true$, where $l$ is the index of the leader of view $j$.

Consider a \BCname\ with identification $id$.
For simplicity of exposition, we denote the delivery
\textbf{key}$(\lr{id,m})$ by $deliver_2(\lr{id,m}$,
\textbf{lock}$(\lr{id,m})$ by $deliver_3(\lr{id, m}$, and
\textbf{commit}$(\lr{id,m})$ by $deliver_4(\lr{id, m}$.
In addition, the variables $D_{key}, D_{lock},D_{commit}$ from
Algorithm~\ref{alg:BAoperation} are called here $D_2,D_3,D_4$,
respectively. 
We say that a \BCname\ with identification $id$
\emph{completes} when $f+1$ honest parties $deliver_4(\lr{id, \lr{v,\sigma}})$ for some $\sigma$ and $v$, and its \emph{completion-proof} is a string
$\sigma'$, which satisfies $\emph{threshold-validate}(\lr{\lr{id,4},v},\sigma') = \emph{true}$.

\paragraph{Properties of the \BCname.}

The next four lemmas prove important properties of the \BCname.

\begin{lemma}
\label{lem:provability1}

Consider a \BCname\ with identification $id$.
 For every value $v$ and $j \in \{2,3,4\}$,
  if some honest party gets $\sigma$ s.t.
  $\emph{threshold-validate}(\lr{\lr{id,j},v},\sigma) = \emph{true}$,
  then at least $f+1$ honest parties previously $deliver_j(id,\langle
  v, \sigma' \rangle)$ for some $\sigma'$.

\end{lemma}

\begin{proof}

By the second part of the \emph{PB-provability} property of the
\SBCname, at least $f+1$ honest parties
$deliver(\lr{id,j},\lr{v,\lr{\sigma_{ex},\sigma_{in}}})$ with some
$\lr{\sigma_{ex},\sigma_{in}}$ in the $j^{th}$ instance of the \SBCname.
Therefore, by the code, at least $f+1$ honest parties previously
$delivers_{j}(id,\langle v, \sigma_{in} \rangle)$.

\end{proof}

\begin{lemma}
\label{lem:provability2}

Consider a \BCname\ with identification $id$.
For every value $v$ and $j \in \{2,3,4\}$, if some honest party
$deliver_j(id,\langle v, \sigma \rangle)$, than \\
$\emph{threshold-validate}(\lr{\lr{id,j-1},v},\sigma) =
\emph{true}$.

\end{lemma}
  
\begin{proof}

Let $p_i$ be an honest party that
$delivers_j(id,\lr{v,\sigma_{in}})$ for some $v$, and $j \in
\{2,\ldots,4\}$.
By the code, $p_i$
$delivers(\lr{id,j},\lr{v,\lr{\sigma_{ex},\sigma_{in}} })$ in an
\SBCname\ with identification $\lr{id,j}$ and some $\sigma_{ex}$, and
by its \emph{validity} property,
\emph{ex-sbc-validation}$_i(\lr{id,j},\lr{v,\lr{\sigma_{ex},\sigma_{in}}})
= true$.
Therefore, since $j > 1$, we get by the code of
\emph{ex-sbc-validation} that
\emph{threshold-validate}$(\lr{\lr{id,j-1}, v},\sigma_{in})$ =
\emph{true}.

\end{proof}

\begin{lemma}
\label{lem:provability3}

Consider a \BCname\ with identification $id$, two honest parties
$p_1,p_2$, and two values $v_1,v_2$.
For every $j_1,j_2 \in \{2,3,4\}$, if $p_1$ gets $\sigma_1$ s.t.\\
$\emph{threshold-validate}(\lr{\lr{id,j_1},v_1},\sigma_1) =
\emph{true}$ and $p_2$ gets $\sigma_2$ s.t.
$\emph{threshold-validate}(\lr{\lr{id,j_2},v_2},\sigma_2) =
\emph{true}$, then $v_1 = v_2$.
  
\end{lemma}

\begin{proof}

Consider three cases:

\begin{itemize}
  
  \item $j_1 = j_2$. The lemma follows from the first part of the
  provability property of the \SBCname.
  
  \item $j_1 = j_2 +1$. Since $p_1$ gets $\sigma_1$ s.t.
$\emph{threshold-validate}(\lr{\lr{id,j_1},v_1},\sigma_1) =
\emph{true}$, by the second part of the provability property of the
\SBCname, we get that at least one honest party delivers
$\lr{v_1,\lr{\sigma_{ex},\sigma_{in}}}$ for some
$\lr{\sigma_{ex},\sigma_{in}}$ in the $j_2^{th}$ instance of the
\SBCname.
By the external validity property of \SBCname, we get that
$\emph{threshold-validate}(\lr{\lr{id,j_2},v_1},\sigma_{in}) =
\emph{true}$.
Therefore, by the first part of the provability property of \SBCname,
we get that $v_1 = v_2$.

\item $j_1 = j_2 +2$. Follows by transitivity.  

\end{itemize}

\end{proof}

\begin{lemma}
\label{lem:BCtermination}

Consider a \BCname\ instance with identification $id$.
If the sender is honest, no honest party invokes abandon(id), all
messages among honest parties arrived, and the message $m$ that is
being broadcast is externally valid, 
then the broadcast completes and returns a completion-proof to the
sender. 

\end{lemma}

\begin{proof}

Consider a \BC{id}{v} invocation by party $p_i$.
By the code of \emph{ex-sbc-validation}, any message is externally
valid for the the first instance of an \SBCname\ consisting it.
Therefore, by the PB-termination property of \SBCname, $p_i$ gets
$\sigma_1$ that satisfies\\ \emph{threshold-validate}$(\lr{\lr{id,1},
v},\sigma_1)$ = \emph{true}, and thus the conditions of the
PB-termination property is satisfied in the second instance of
\SBCname\ as well ($\lr{v,\sigma_1}$ is externally valid).
So by the PB-termination property of \SBCname again, (1) all honest
party $deliver_2(\lr{id,\lr{v,\sigma_1}}$; and (2) $p_i$
gets $\sigma_2$ that satisfies
\emph{threshold-validate}$(\lr{\lr{id,2}, v},\sigma_2)$ =
\emph{true}.
The lemma follows by applying the same arguments for the remain two
\SBCname\ instances.

\end{proof}

\paragraph{Key protocol property.}

The next Lemma use two of the \BCname\ properties proven above to
prove a key property for both safety and progress of the protocol.

\begin{lemma}
\label{lem:f1causality}

Consider a view $j$ and let $p_l$ be the chosen leader of view $j$.
For every $i \in \{2,3\}$, if an honest party gets a 
``view-change'' message that includes $\lr{v,\sigma}$ s.t.\\ 
$\emph{threshold-validate}(\lr{\lr{\lr{id,l,j},i},v}, \sigma)
= true$, then all honest parties get a 
``view-change'' message that includes $\lr{v,\sigma'}$ s.t.\ 
$\emph{threshold-validate}(\lr{\lr{\lr{id,l,j},i-1},v}, \sigma')
= true$.

\end{lemma}

\begin{proof}

By Lemma~\ref{lem:provability1}, there exists a set $P \in \Pi_h$ of
at least $f+1$ honest parties that \\
$deliver_i(\lr{id,l,j},\lr{v,\sigma'})$ for some $\sigma'$, and
by Lemma~\ref{lem:provability2},
$\emph{threshold-validate}(\lr{\lr{\lr{id,l,j},i-1},v}, \sigma') =
true$.
By the code, the parties in $P$ set their $D_i[l,j] = \lr{v,\sigma'}$,
and by the PB-abandon-ability property of the \SBCname, it happened
before they invoke $abandon(\lr{id,l,j})$, and thus before they send their
``view-change'' messages.
Since every honest party waits to receive \emph{view-change} messages
from $2f+1 = n - f$ different parties, we get that every honest party
gets a view-change message that includes $\lr{v,\sigma'}$ s.t.\ 
$\emph{threshold-validate}(\lr{\lr{\lr{id,l,j},i-1},v}, \sigma')
= true$.

\end{proof}

\paragraph{Agreement proof.}

\begin{lemma}
\label{lem:lockonj}

If a party $p$ decide on a value $v$ in a view $j$, then the
{LOCK} variables of all honest parties are at least $j$ when they
move to view $j+1$.

\end{lemma}

\begin{proof}

Let $p_l$ be the leader of view $j$.
By the code, since $p$ decides $v$ in view $j$, we know that $p$
received a ``view-change'' message in view $j$ with $\lr{v,\sigma}$
s.t.\ $\emph{threshold-validate}(\lr{\lr{\lr{id,l,j},3},v},
\sigma) = true$.
By Lemma~\ref{lem:f1causality}, we get that all honest parties
received  a  ``view-change'' message that includes $\lr{v,\sigma'}$
s.t.\ $\emph{threshold-validate}(\lr{\lr{\lr{id,l,j},2},v}, \sigma')
= true$.
Therefore, by the code, all honest parties set their \emph{LOCK}
variables to $j$ ( if it was smaller than $j$).

\end{proof}

\noindent The next corollary follows from Lemma~\ref{lem:lockonj} and
the fact that the \emph{LOCK} variables are never decreased.  

\begin{corollary}
\label{col:lockj}

If a party $p$ decides on a value $v$ in a view $j$, then the
{LOCK} variables of all honest parties are at least $j$ when they
start any view $j' > j$.

\end{corollary}

\begin{lemma}
\label{lem:agreement1}

If an honest party $p$ decides on $v$ in view $j$, than all honest
parties that decide in view $j$, decide $v$ as well.

\end{lemma}

\begin{proof}

Let $p_l$ be the leader of view $j$.
By the code, an honest party decides $v$ in view $j$ if and only if
it gets a ``view-change'' message with $\lr{v,\sigma}$ s.t.\
$\emph{threshold-validate}(\lr{\lr{\lr{id,l,j},3},v}, \sigma) = true$.
The lemma follows from the first part of the PB-provability property
of the third instance of the \SBCname.

\end{proof}

The next lemma shows that all \textbf{keys} for views higher than or
equal to a view in which a decision on a value $v$ was made are
valid only for $v$.

\begin{lemma}
\label{lem:keyinduction}

Assume an honest party $p$ decides on $v$ in a view $j$.
Than for every view $j' \geq j$, and every honest party that
gets $\lr{v',\sigma'}$ s.t.\
$\emph{threshold-validate}(\lr{\lr{\lr{id,l',j'},1},v'}, \sigma') =
true$, where $p_{l'}$ is the leader of view $j'$, we get that
$v' = v$. 

\end{lemma}

\begin{proof}

We prove by induction on the view number $j'$.

\textbf{Base:} $j' = j$.
By the code, $p$ gets a ``view-change'' message with $\lr{v,\sigma}$
s.t.\\ $\emph{threshold-validate}(\lr{\lr{\lr{id,l',j'},3},v}, \sigma)
= true$.
Therefore, by Lemma~\ref{lem:provability3}, we get that if an honest
party gets $\lr{v',\sigma'}$ s.t.\
$\emph{threshold-validate}(\lr{\lr{\lr{id,l',j'},1},v'}, \sigma') =
true$, then $v' = v$.

\textbf{step:} Assume the lemma holds for all $j''$, $j \leq j'' \leq
j'$, we now show that it holds for $j' + 1$ as well.
Assume by a way of contradiction that some honest party gets
$\lr{v',\sigma'}$ s.t.\
$\emph{threshold-validate}(\lr{\lr{\lr{id,l'+1,j'+1},1},v'}, \sigma')
= true$ and $v' \neq v$.
By the second part of the PB-provability property, at least one honest
party $p_1$ delivers $\lr{v',\lr{\sigma_{ex},\sigma_{in}}}$ for some
$\lr{\sigma_{ex},\sigma_{in}}$ in the \SBCname\ instance with
identification $\lr{\lr{id,l'+1,j'+1},1}$.
By Corollary~\ref{col:lockj}, the \emph{LOCK} variable of $p_1$ at
the beginning of view $j' + 1$ is at least $j$.
Therefore, by the PB-validity property of the \SBCname\ and the
definition of its external validation for $p_1$, we get that
there is a view $j''$, $j \leq j'' \leq j'$ s.t.\
$\emph{threshold-validate}(\lr{\lr{\lr{id,l'',j''},1},v'},
\sigma_{in}) = true$.
A contradiction to the inductive assumption.

\end{proof}

\begin{lemma}

Algorithms~\ref{alg:BAoperation} and~\ref{alg:BAmessages} satisfy
agreement.

\end{lemma}

\begin{proof}

Let party $p$ be the first to decide, let $j$ be the view in
which $p$ decides, and let $v$ be the value $p$ decides on.
First, by Lemma~\ref{lem:agreement1}, all honest parties that decide
in view $j$, decides $v$.
Now consider view $j' > j$.
By Lemma~\ref{lem:keyinduction}, if an honest party gets
$\lr{v',\sigma'}$ s.t.\ $\emph{threshold-validate}(\lr{\lr{\lr{id,l',j'},1},v'}, \sigma') =
true$, where $p_{l'}$ is the leader of view $j'$, then $v' = v$.
Thus, by Lemma~\ref{lem:provability3}, if an honest party gets
$\lr{v',\sigma'}$ s.t.\
$\emph{threshold-validate}(\lr{\lr{\lr{id,l',j'},3},v'}, \sigma') =
true$, where $p_l'$ is the leader of view $j'$, then $v' = v$.
Therefore, since an honest party decides on $v'$ in view $j'$ only if
it gets a ``view-change'' message with $\lr{v',\sigma'}$ s.t.\\
$\emph{threshold-validate}(\lr{\lr{\lr{id,l,j},3},v}, \sigma') =
true$, we get that if an honest party decide $v'$ in view $j'$, then
$v' = v$.

\end{proof}

\paragraph{Termination and validity proofs.}

\begin{lemma}
\label{lem:lock-key}

Consider a view $j$ and assume that some honest party $p$ sets its
\textit{LOCK} variable to $j$ in view $j$, then all honest parties
set their \textit{KEY} variable to
$\lr{j,v,\sigma}$ s.t.\ $\textbf{key} = \lr{j,\sigma}$ is valid for
$v$ before moving to view $j+1$.

\end{lemma}

\begin{proof}

Let $p_l$ be the leader of view $j$.
By the code, since $p$ sets its \textit{LOCK} variable to $j$ in view
$j$, we know that $p$ received a ``view-change'' message in view $j$
with $\lr{v,\sigma}$ s.t.\\
$\emph{threshold-validate}(\lr{\lr{\lr{id,l,j},2},v}, \sigma) =
true$.
By Lemma~\ref{lem:f1causality}, we get that all honest parties
received  a  ``view-change'' message that includes $\lr{v,\sigma'}$
s.t.\ $\emph{threshold-validate}(\lr{\lr{\lr{id,l,j},1},v}, \sigma')
= true$.
Therefore, by the code, all honest parties set their \textit{KEY}
variables to $\lr{j,v,\sigma}$ s.t.\ $\textbf{key} = \lr{j,\sigma}$ is
valid for $v$.

\end{proof}

\begin{lemma}
\label{lem:messagevalid}

If all honest parties start with externally valid values for byzantine
agreement, then for every view $j \geq 1$, all messages broadcast
by honest parties are externally valid for the \BCname.

\end{lemma}

\begin{proof}

First consider view $j=1$.
Since all honest parties start with externally valid values for
byzantine agreement, and since by the code of \emph{ex-bc-validation}
(lines~\ref{line:BAvalidty_starts}--\ref{line:BAvalidty_ends} in
Algorithm~\ref{alg:BAoperation}), a valid key is not required in view
1, we get that all messages broadcast
by honest parties in view 1 are externally valid for the \BCname.

Now consider a view $j>1$.
Let $i = max(\{k \mid \text{there is an honest party $p$ who's }
\textit{lock} = k \text{ when it begins view $j$} \})$.
By the code of the external validity for \BCname, we need to show
that when a party begins view $j$, its \textit{KEY} variable is
equal to $\lr{k,v,\sigma}$ such that:
\begin{enumerate}
  
  \item $\lr{k,\sigma}$ is a valid \textbf{key} for $v$,
  
\item $k \geq i$, and

\item $v$ is an external valid value for byzantine agreement.
  
\end{enumerate}
By Lemma~\ref{lem:lock-key}, all honest parties set
their \textit{KEY} variable to $\lr{i,v,\sigma}$ when they end view
$i$ s.t.\ $\textbf{key} = \lr{j,\sigma}$ is valid for $v$.
By the code, after view $i$, the \textit{KEY} variable is updated to
$\lr{k,v',\sigma'}$ by an honest party $p$ only if $p$ receives a
``view-change'' message with $\lr{v',\lr{k,\sigma'}}$ s.t.\ $k > i$
and $\lr{k,\sigma'}$ is a valid \textbf{key} for $v'$.
Therefore, we get that (1) and (2) are satisfied.
Now by the PB-provability property of \SBCname, we get that at least
one honest party delivers $\lr{v,\sigma''}$ for some $\sigma''$ at
some instance of \SBCname, and thus by the code of ex-sbc-validation
and ex-bc-validation functions, $v$ is an external valid value for
byzantine agreement.

\end{proof}

The following corollary follows immediately from
Lemma~\ref{lem:messagevalid}.

\begin{corollary}

Algorithms~\ref{alg:BAoperation} and~\ref{alg:BAmessages} satisfy
validity.

\end{corollary}

The following observation follows from the fact that honest parties
echo ``skip'' messages before setting $skip = true$.

\begin{observation}
\label{obs:allskip}

Consider a view $j$, which all honest parties start.
If all messages sent by honest parties in a view $j$ arrived and
some honest party sets $skip = true$ in view $j$, then all honest
parties set $skip = true$ in view $j$.

\end{observation}

\begin{lemma}
\label{lem:honestcomplete}

Assume that all honest parties start with externally valid values for
byzantine agreement.
Consider a view $j$, which all honest parties start.
If all messages sent by honest parties in view $j$ arrived and no
honest party sets $skip = true$ in view $j$,
then all \BCname s issued by honest parties have completed and
returned completion-proofs.

\end{lemma}

\begin{proof}

Since no honest party sets $skip = true$ in view $j$, by the
code, no honest party invokes \emph{abandon} in a
the \BCname s of view $j$.
By Lemma~\ref{lem:messagevalid}, all messages broadcast by honest
parties are externally valid for the \BCname.
Therefore, by Lemma~\ref{lem:BCtermination}, all broadcasts issued by
honest parties have completed and returned completion-proofs.

\end{proof}

\begin{lemma}
\label{lem:allsetskip}

Consider a view $j$, which all honest parties start. 
If all messages sent among honest parties in a view $j$ arrived and
all \BCname s issued by honest parties in view $j$ have been completed
and returned completion-proofs, then all honest parties set $skip =
true$ in view $j$.

\end{lemma}

\begin{proof}

By the code, all hones parties sent a ``done'' message with a
completion-proof (of their broadcast) to all other parties.
Therefore, since all messages sent among honest parties in a view $j$
arrived, we get that all honest parties sent a ``share-skip'' message
to all other parties.
Since there are at least $2f+1$ honest parties, all honest parties
computed a valid threshold signature on ``skip'' and sent it to all
parties.
Thus, all honest parties gets a ``skip'' message with a valid
threshold signature and set $skip = true$ in view $j$.

\end{proof}

\begin{lemma}
\label{lem:roundcomplete}

Assume all honest parties start with externally valid values for
byzantine agreement, and consider a view $j$, which all
honest parties start. 
If all messages sent among honest parties in a view $j$ arrived, then
view $j$ completed.

\end{lemma}

\begin{proof}

We first show that all honest parties set $skip = true$ in
view $j$.
By Observation~\ref{obs:allskip}, if some honest party sets $skip =
true$ in view $j$, then all honest parties set $skip = true$ in
view $j$.
Now assume by a way of contradiction that no honest party sets $skip
= true$ in view $j$.
Therefore, by Lemma~\ref{lem:honestcomplete}, all \BCname s issued by
honest parties have completed and returned completion-proofs, and
thus, by Lemma~\ref{lem:allsetskip}, all honest party set $skip =
true$ in view $j$.
A contradiction.

It remains to show that no honest party waits forever for $elect(j)$
to return or to a ``view-change'' message from $2f+1$ different parties.
Since, all honest parties set $skip =
true$ in view $j$, we get by the code that they all invoke
$elect(j)$, and thus by the termination property of the leader
election abstraction we get that all invocations of $elect(j)$ by
honest parties return.
Therefore, all honest parties send a ``view-change'' messages to all
other parties, and again, since all messages sent among honest
parties in a view $j$ arrived, we get that all honest parties get
``view-change'' messages from at least $2f+1$ different parties.
Hence, all parties move to view $j+1$.

\end{proof}

\begin{lemma}

Algorithms~\ref{alg:BAoperation} and~\ref{alg:BAmessages} satisfy
termination.

\end{lemma}

\begin{proof}

By applying inductively Lemma~\ref{lem:roundcomplete}, we get that the
number of views in the protocol is unbounded, and thus honest parties
never stop sending new messages.
Therefore, termination trivially follows.


\end{proof}

\paragraph{The secret sauce and the efficiency proof.}

\begin{lemma}
\label{lem:2f1complete}

Consider a view $j$ of the protocol with identification
$id$.
If an honest party invokes $elect(\lr{id,j})$, then at least $2f+1$
\BCname s have completed before.

\end{lemma}

\begin{proof}

By the code, if an honest party invokes $elect(\lr{id,j})$, then it
sets its  $skip_j = true$ before.
Thus it got a ``skip'' message with a valid threshold signature.
Thus, some party got $2f=1$ ``skip-share'' messages, and thus at least
one honest party $p$ sent a ``skip-share'' message.
Therefore, by the code again, there is a set $P$ of $2f+1$ parties
s.t.\ for every $p_i \in P$, $p$ receives $\lr{v,\sigma}$ s.t.\ $
\emph{threshold-validate}(\lr{\lr{\lr{id,i,j},4},v}, \sigma)$.
Therefore, by the PB-provability property, for every $p_i \in P$, at
least $f+1$ honest parties deliver a message $m = \lr{v,
\lr{\lr{\sigma_{ex},\sigma_{in}}}}$ for some
$\lr{\sigma_{ex},\sigma_{in}}$ in the $4^{th}$ \SBCname\ instance of
the \BCname\ issued by $p_i$.
Hence, for every $p_i \in P$, at least $f+1$ honest parties
$deliver_4(\lr{id,i,j},\lr{v,\sigma_{in}})$ , and thus we get that at
least $2f+1$ \BCname s complete in view $j$.

\end{proof}

\begin{lemma}
\label{lem:alldecide}

Consider a completed view $j$ of the protocol with identification
$id$.
If the broadcast issued by the leader of view $j$ completed before
the leader was elected, then all honest parties decide in view $j$.

\end{lemma}

\begin{proof}

Let $p_l$ be the leader of view $j$.
Since the broadcast issued by $p_l$ completed before it was elected,
we get that at least $f+1$ honest parties $deliver_4$ a message
$\lr{v,\sigma}$ in $p_l$'s broadcast before sending their
``view-change'' massages.
Thus, since honest parties wait for $2f+1$ ``view-change'', every
honest party gets a ``view-change'' with $\lr{v,\sigma}$ s.t.\ 
\emph{threshold-validate}$(\lr{\lr{\lr{id,l,j},3},v}, \sigma) =
true$.
Therefore, by the code, all honest parties decide in view $j$.

\end{proof}

\begin{lemma}
\label{lem:probabilitytodecide}

Assume all honest parties start with externally valid values for
byzantine agreement, and consider a completed view $j$.
Then there is a
probability of at least $\frac{2}{3}$ that all honest parties decided in view
$j$.

\end{lemma}

\begin{proof}

Since view $j$ is completed at least $f+1$ honest parties invoked
$elect(j)$, and let $p_l$ be the leader of view $j$.
By Lemma~\ref{lem:2f1complete}, at least $2f+1$ \BCname s have
completed before the first honest party invokes $elect(j)$.
Therefore, by the validity and unpredictability of the leader election
abstraction, we get that the \BCname\ issued by the leader $p_l$
completed before it was elected with a probability of
$\frac{2f+1}{3f+1} > \frac{2}{3}$.
Therefore, by Lemma~\ref{lem:alldecide}, all honest parties
decide in view $j$ with probability of at least $\frac{2}{3}$.

\end{proof}

\begin{observation}
\label{obs:norash}

An honest party cannot move to view $j+2$ before view $j$ was
completed.

\end{observation}

\begin{lemma}

Algorithms~\ref{alg:BAoperation} and~\ref{alg:BAmessages} satisfy
efficiency.

\end{lemma}

\begin{proof}

First note that since the number of messages sent by honest parties in
an \SBCname is at most $2n$, we get that every view has $O(n^2)$
messages sent by honest parties.
Denote this number by $T$, and let $X$ be the total number of messages
sent by honest parties in the protocol before all honest parties
decide.
By Observation~\ref{obs:norash}, $Pr[X > lT]$ is equal to the
probability that some honest parties do not decide in the first $l$
views.
By Lemma~\ref{lem:probabilitytodecide}, the probability of this is
less than $(\frac{1}{3})^{l}$, and the lemma follows.

\end{proof}

\section{Lower Bound on the number of messages for Asynchronous
Byzantine Agreement With Adaptive Adversary}
\label{app:LB}

A recent theorem of Abraham et al. \cite{ADDNR18} provides a lower
bound for \textit{Synchronous} Byzantine Agreement against a
\textit{strongly rushing adaptive adversary}:

\begin{theorem}[ADDNR \cite{ADDNR18}]\label{thm:ADDNR18}
If a protocol solves Synchronous Byzantine broadcast with $(1/2) + \epsilon$ probability against a strongly rushing adaptive adversary, 
then in expectation, honest parties collectively need to send at
least $(\epsilon f /2)^2$ messages. 
\end{theorem}

The \emph{strongly rushing adaptive adversary} assumed in Abraham et
al.~\cite{ADDNR18} can adaptively decide which $f$ parties to corrupt
and when to corrupt them.
In particular, the adversary is allowed to decide to corrupt a party
$p$ after observing the messages sent by $p$ in round $r$.
In addition, after corrupting $p$, the adversary can remove $p$'s
round-$r$ messages from the network before they reach other honest
parties\footnote{In comparison, a standard adaptive adversary (in the synchronous
model) cannot ``take back'' or remove $p$'s round-$r$ messages to
other honest parties.}.

Note that while Theorem \ref{thm:ADDNR18} is proven for Byzantine
broadcast, it immediately induces a lower bound for binary Byzantine
agreement.
We now use it to prove the following lower bound

\begin{theorem}
If a protocol solves binary Byzantine
agreement with $(1/2) + \epsilon$ probability against an asynchronous adaptive adversary, then in
expectation, honest parties collectively need to send at least $(\epsilon f /2)^2$ messages. 
\end{theorem}

\begin{proof}

If a protocol sends at most $(\epsilon f /2)^2$ messages against any asynchronous adaptive adversary
then it will clearly send at most  $(\epsilon f /2)^2$ messages against any restricted adversary 
that works in the synchronous model of communication (where any message sent in round $r$ by an honest party
must arrive by the end of round $r$).

The core observation is that the asynchronous adaptive adversary has the ability to
not deliver messages of parties it corrupts. This follows from the 
definition of Termination - the protocol must terminate
even if some messages sent by parties that are corrupted
will never be delivered. This is true even if some of these messages have been sent
before the adversary decided to corrupt this party.

Hence even when we restrict the asynchronous adaptive adversary to a synchronous message passing model
the adaptive adversary can still essentially remove messages of a party $p$ sent in round $r$
where $r$ is the round that the adversary decided to corrupt party $p$.
Therefore the adaptive asynchronous adversary can fully simulate the synchronous strongly rushing adaptive adversary.
In particular, the adaptive asynchronous adversary can simulate the behaviour as in Theorem \ref{thm:ADDNR18} and cause 
the protocol to have an $1/2-\epsilon$ error probability.
\end{proof}

\end{appendices}

\newpage
\bibliographystyle{plain}
\bibliography{bibliography}

\end{document}